\documentclass{easychair}

\usepackage{stmaryrd}
\usepackage{amssymb}

\input{cpr-macros-iwc2017}

\newtheorem{definition}{Definition}
\newtheorem{example}{Example}
\newtheorem{lemma}{Lemma}
\newtheorem{theorem}{Theorem}

\title{Critical Peaks Redefined\\
$\clseq{\clslub{\aStp}{\bStp}}{\clstop}$}
\author{Nao Hirokawa\inst{1}
\and
Julian Nagele\inst{2}
\and
Vincent van Oostrom\inst{2}
\and
Michio Oyamaguchi\inst{3}}
\institute{
  School of Information Science, JAIST, Japan\\
  \email{hirokawa@jaist.ac.jp}
\and
  Department of Computer Science, University of Innsbruck, Austria\\
  \email{julian.nagele@uibk.ac.at}, \email{Vincent.van-Oostrom@uibk.ac.at}
\and
  Nagoya University, Japan\\
  \email{oyamaguchi@za.ztv.ne.jp}
}
\titlerunning{Critical Peaks Redefined:
$\clseq{\clslub{\aStp}{\bStp}}{\clstop}$}
\authorrunning{Hirokawa, Nagele, van Oostrom, and Oyamaguchi}

\begin{document}

\maketitle

\begin{abstract}
  Let a \emph{cluster} be a term with a number of patterns occurring in it.
  We give two accounts of clusters,
  a \emph{geometric} one as sets of (node and edge) positions, and
  an \emph{inductive} one as pairs of terms with gaps ($2$nd order variables)
  and pattern-substitutions for the gaps.
  We show both notions of cluster and the corresponding 
  refinement/coarsening orders on them, to be isomorphic.
  This equips clusters with a lattice structure which
  we lift to (parallel/multi) steps to yield 
  an alternative account of the notion 
  of critical peak.
\end{abstract}
    
\section{Introduction}

The critical pair lemma~\cite{Huet:80} is the cornerstone for proving 
confluence of first-order term rewrite systems. 
It expresses that a term rewrite system is locally confluent if and only if all
its critical pairs are joinable. In case the system is moreover terminating this
allows to reduce, by Newman's Lemma~\cite{Newm:42}, checking confluence to
checking joinability of its critical pairs, which are finitely many in the 
case of a finite term rewrite system. This forms the basis for Knuth--Bendix completion.
The termination condition cannot be omitted without more from the critical pair lemma:
On the one hand, a non-terminating TRS may fail to be confluent even in the
absence of critical pairs due to non-left-linearity, as established by Klop.
On the other hand, a non-terminating left-linear TRS may still fail to 
be confluent despite that all its critical pairs are joinable.
Still, for orthogonal, i.e.\ left-linear and without critical pairs, TRSs
confluence does hold for \emph{geometric} reasons:
redex-patterns can be contracted independently 
of each other, inducing a notion of residual.
Starting with Church and Rosser a rich theory of residuals
has been developed, but comparatively little attention has been
paid to the result that lies at its basis:
a strengthening of the critical pair lemma stating
that any peak \emph{either} is (a variable-instance of) a 
critical peak \emph{or} can be decomposed into smaller peaks.
We present such a \emph{critical peak }lemma.

Since both for defining rewriting and for defining critical peaks
the notion of \emph{encompassment} is essential, we start off with analysing it.
In particular, we call a term with a number of patterns (think of left-hand sides of rules)
encompassed by it a \emph{cluster}, and introduce two representations of clusters:
a \emph{geometric} one as sets of (node and edge) positions, and
an \emph{inductive} one as pairs of terms with gaps ($2$nd order variables)
and pattern-substitutions for the gaps.
One can think of these two representations as corresponding to
the pictures respectively the formal proof of the critical pair
lemma as found in e.g.~\cite{Baad:Nipk:98,Ohle:02,Tere:03}.
Here we give formal accounts of \emph{both} and of the 
refinement/coarsening order on them, and show them to be isomorphic.
This allows one to bridge the gap between the often informal
geometric intuition (`proofs by picture') at the basis of properties
of residuals, and the inductive nature of (`terms and steps') of term 
rewriting.
As a first example (we anticipate many more) we redefine 
in this paper the  notion of critical peak in a purely lattice theoretic way, 
based just on the coarsening/refinement order of clusters.
More precisely, we call a local peak between steps $\astp$ and $\bstp$
\emph{critical} if $\clseq{\clslub{\astp}{\bstp}}{\clstop}$,
or in words, the redex-patterns of $\astp$ and $\bstp$ must be
overlapping (if not, $\clslt{\clslub{\astp}{\bstp}}{\clstop}$ 
as the join would `miss an edge' and comprise \emph{two} patterns)
and all symbols in the source of the peak must be
part of either pattern (if not, $\clslt{\clslub{\astp}{\bstp}}{\clstop}$
as the join would `miss a node' and not comprise \emph{all} symbols).
This definition captures the intuition behind something
being \emph{critical}: the source does encompass both redex-patterns
but only just so, nothing else is encompassed.

\textbf{We restrict ourselves to first-order term rewriting
and to linear patterns.}

\section{Clusters and the refinement lattice}

We introduce clusters and the refinement order on them.
It is standard to represent terms as labelled trees
using sequences of positive natural numbers called \emph{positions}
for the nodes of the tree.
Following~\cite[Chapter~8]{Tere:03} we extend this 
by having both \emph{vertex} and \emph{edge} positions,
see Figure~\ref{fig:vertexedge}{(a)}.
Note that $\avtxidx$ and $\aedgidx$ stand for
$\possco{\aidx}{\vposemp}$ and $\possco{\aidx}{\eposemp}$, respectively.
Let $\possco{\aidx}{\aPos}$ denote
$\{ \possco{\aidx}{\apos} \mid \setin{\apos}{\aPos} \}$.
\begin{figure}
\begin{minipage}[b]{.31\textwidth}
\def\afig{$\possco{\posone}{\possco{\posone}{\vposone}}$}
\def\bfig{$\possco{\posone}{\possco{\posone}{\eposone}}$}
\def\cfig{$\possco{\posone}{\vposone}$}
\def\dfig{$\possco{\posone}{\vpostwo}$}
\def\efig{$\possco{\posone}{\eposone}$}
\def\ffig{$\possco{\posone}{\epostwo}$}
\def\gfig{$\eposemp$}
\def\hfig{$\eposone$}
\def\ifig{$\vposemp$}
\def\jfig{$\vposone$}
 \centering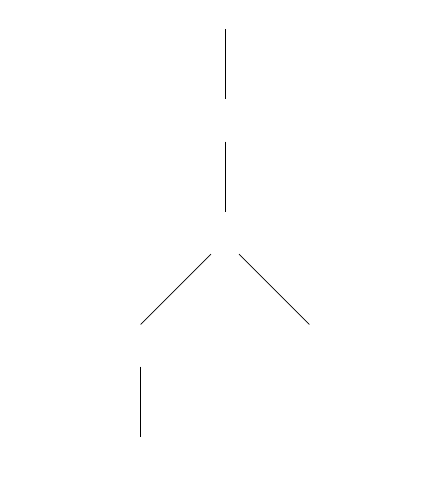
(a) Vertex/edge positions~\cite{Tere:03}
\end{minipage}%
\begin{minipage}[b]{.69\textwidth}
\[
\begin{array}{@{}l@{~}l@{}l@{}}
\hline
\rule{0mm}{1.3em}
\emptyset
&
\pair{\syma{\symb{\symc{\bsymzer},\asymzer}}
     }{& \llbracket\,\rrbracket}
= \clsbot
\\[0.3em]
\setstr{\vposemp}
&
\pair{\mtrmvara{\symb{\symc{\bsymzer},\asymzer}}
     }{& \msubstr{\amtrmvar}{\syma{\trmvar{\natone}}}
      }
\\[0.3em]
\setstr{\vposone,\possco{\posone}{\vposone}}
&
\pair{\syma{\mtrmvara{\mtrmvarb{\bsymzer},\asymzer}}
     }{& \msubstr{\amtrmvar,\bmtrmvar}{\symb{\trmvar{\natone},\trmvar{\nattwo}},
                                     \symc{\trmvar{\natone}}}
      }
\\[0.3em]
\setstr{\vposone,\possco{\posone}{\eposone},\possco{\posone}{\vposone}}
&
\pair{\syma{\mtrmvara{\bsymzer,\asymzer}}
     }{& \msubstr{\amtrmvar}{\symb{\symc{\trmvar{\natone}},\trmvar{\nattwo}}}
      }
\\[0.3em]
\setstr{\vposone,\possco{\posone}{\epostwo},\possco{\posone}{\vpostwo}}
&
\pair{\syma{\mtrmvara{\symc{\bsymzer}}}
     }{& \msubstr{\amtrmvar}{\symb{\trmvar{\natone},\asymzer}}
      }
\\[0.3em]
\left\{
\begin{matrix}
\vposemp,
\eposone,
\vposone,
\possco{\posone}{\eposone},
\possco{\posone}{\vposone},
\\
\possco{\posone}{\possco{\posone}{\eposone}},
\possco{\posone}{\possco{\posone}{\vposone}},
\possco{\posone}{\epostwo},
\possco{\posone}{\vpostwo}
\\
\end{matrix}
\right\}
&
\pair{\amtrmvar}
     {& \msubstr{\amtrmvar}{\syma{\symb{\symc{\bsymzer},\asymzer}}}}
= \clstop
\\[1em]
\hline
\end{array}
\]
\centering (b) Geometric and inductive clusters.
\end{minipage}%
  \caption{Positions and clusters for $\syma{\symb{\symc{\bsymzer},\asymzer}}$.}
  \label{fig:vertexedge}
\end{figure}

\begin{definition}
  The $\salgtre$ algebra has as carrier sets of positions, and interpretations
  \[ \funap{\op{\asym}{\salgtre}}{\vec{\aPos}} \isdefd 
       \setlub{\setstr{\edgpos{\posemp},\vtxpos{\posemp}}
             }{\setLub{\aidx}{\possco{\aidx}{\aiPos{\aidx}}}
             } \]
\end{definition}
The set of all positions of a term arises by assigning 
$\setstr{\edgpos{\posemp},\vtxpos{\posemp}}$ to variables.
We are interested in the \emph{internal} positions arising from
assigning $\setemp$ to variables and removing the root edge $\edgpos{\posemp}$.
\begin{definition}
  A \emph{cluster} for a given term $\atrm$
  is a set of internal positions of $\atrm$ such that
  if an edge $\possco{\apos}{\aedgidx}$ is in the set,
  its endpoints $\possco{\apos}{\vtxpos{\posemp}}$ and $\possco{\apos}{\avtxidx}$
  are too.
  Its connected components are called \emph{patterns}.
\end{definition}
Our patterns correspond to those in~\cite[Chapter~8]{Tere:03}.
\begin{example} \label{exa:gcluster}
  The first column of Figure~\ref{fig:vertexedge}(b) lists
  some clusters for the term $\syma{\symb{\symc{\bsymzer},\asymzer}}$,
  of which the first and third are not patterns (for the latter,
  $\possco{\posone}{\eposone}$ is missing).
  Since $\setstr{\vposone,\possco{\posone}{\eposone}}$ 
  lacks the endpoint $\possco{\posone}{\vposone}$ of
  the edge $\possco{\posone}{\eposone}$,
  it is not a cluster.
  Note that 
  $\setstr{\vposone,\possco{\posone}{\eposone},\possco{\posone}{\vposone}} \neq
   \setstr{\vposone,\possco{\posone}{\vposone}}$;
  the former is a cluster comprising a single pattern, whereas
  the latter comprises two patterns.
\end{example}
\begin{lemma}
  For any term, the clusters for that term constitute a 
  finite distributive lattice with respect to the subset order $\ssetle$.
\end{lemma}
\begin{proof}
  By terms being finite, and
  properties being inherited from the subset order.
  \end{proof}
We give an alternative definition of clusters.
To keep both apart, we will refer to the above notions
as \emph{geometric} and to the ones introduced 
below as \emph{inductive}.
\begin{definition}
  A \emph{skeleton} is constructed from function symbols
  and $1$st and $2$nd order variables, the latter called \emph{gaps}.
  It is a \emph{pattern-}skeleton of \emph{arity} $\anat$,
  if it is not a variable and \emph{standard}: the vector of variables occurring 
  from left to right is $\trmvar{\natone},\ldots,\trmvar{\anat}$.
  A \emph{term} is a skeleton without gaps,
  and a \emph{pattern} is a pattern-skeleton without gaps.
  A \emph{cluster} is a pair $\pair{\amtrm}{\msubstr{\vec{\amtrmvar}}{\vec{\apat}}}$
  with $\amtrm$ a skeleton linear in the gaps,
  and  $\msubstr{\vec{\amtrmvar}}{\vec{\apat}}$ 
  substituting patterns $\vec{\apat}$ for those, respecting arities.
  We say $\pair{\amtrm}{\msubstr{\vec{\amtrmvar}}{\vec{\apat}}}$
  is a cluster \emph{for} the term
$\msubstra{\vec{\amtrmvar}}{\vec{\apat}}{\amtrm}$.
\end{definition}
\begin{example} \label{exa:icluster}
  The inductive clusters corresponding to the six geometric clusters of 
  Example~\ref{exa:gcluster} are
  listed in the second column of Figure~\ref{fig:vertexedge}(b).
  The inequality in that example corresponds to the inequality
  $\clsneq{\pair{\syma{\mtrmvara{\bsymzer,\asymzer}}
               }{\msubstr{\amtrmvar}{\symb{\symc{\trmvar{\natone}},\trmvar{\nattwo}}}
               }
         }{\pair{\syma{\mtrmvara{\mtrmvarb{\bsymzer},\asymzer}}
               }{\msubstr{\amtrmvar,\bmtrmvar}{\symb{\trmvar{\natone},\trmvar{\nattwo}},
                                       \symc{\trmvar{\natone}}}
               }
         }$.
\end{example}
\begin{definition}
  The \emph{coarsening} order $\sclsle$ (the \emph{refinement} order $\sclsge$)
  on inductive clusters is defined by
  $\clsge{\clsstr{\bmtrm}{\bmsub}}{\clsstr{\amtrm}{\amsub}}$ if
  $\trmeq{\msubc{\bmtrm}}{\amtrm}$ and $\subeq{\bmsub}{\subcom{\cmsub}{\amsub}}$ 
  for some pattern-skeleton substitution $\cmsub$.
\end{definition}
\begin{example} \label{exa:refinement}
   For the term $\syma{\asymzer}$ the pattern comprising
   both symbols may be refined to the cluster comprising two patterns:
   $\clsgt{\pair{\cmtrmvar
                }{\msubstr{\cmtrmvar}{\syma{\asymzer}}
                }
          }{\pair{\mtrmvara{\bmtrmvar}
               }{\msubstr{\amtrmvar,\bmtrmvar}{\syma{\trmvar{\natone}},\asymzer}
               }
          }$ witnessed by the pattern-skeleton substitution
   $\msubstr{\cmtrmvar}{\mtrmvara{\bmtrmvar}}$.
   Geometrically this corresponds to
   $\setgt{\setstr{\vposemp,\eposone,\vposone}
         }{\setstr{\vposemp,\vposone}
         }$.
\end{example}
The main result of this section can be viewed as an instance of Birkhoff's Fundamental
Theorem of Finite Distributive Lattices, expressing that \emph{all} such 
lattices can be represented via downward-closed sets of join-irreducible 
elements ordered by subset.
\begin{theorem}
  For a given term, 
  geometric clusters ordered by $\ssetle$ are isomorphic to 
  inductive clusters, up to renaming of gaps, ordered by $\sclsle$.
  The order is a finite distributive lattice.
\end{theorem}
\begin{proof}
  First note that we can map any inductive cluster 
  $\pair{\amtrm}{\msubstr{\vec{\amtrmvar}}{\vec{\apat}}}$ 
  to a geometric cluster by means of what we call a \emph{cluster}
  algebra, a pair of algebras for interpreting both components:
    \begin{itemize}
    \item
      $\salgshi$
      for interpreting $\amtrm$:
      $\funap{\op{\asym}{\salgshi}}{\vec{\aPos}} \isdefd 
       \setLub{\aidx}{\possco{\aidx}{\aiPos{\aidx}}}$; and
    \item
      $\salgtre$
      for interpreting $\vec{\amtrmvar}$ via $\vec{\apat}$:
      $\funap{\op{\asym}{\salgtre}}{\vec{\aPos}} \isdefd 
       \setlub{\setstr{\edgpos{\posemp},\vtxpos{\posemp}}
             }{\funap{\op{\asym}{\salgshi}}{\vec{\aPos}}
             }$, then removing the root edge.
    \end{itemize}
  This map is seen to be a bijection. That it preserves the order is seen:
  \begin{list}{}{}
  \item[(geometric $\Rightarrow$ inductive)]
    by induction on the term, simultaneously building the inductive clusters
    from both geometric clusters and the witnessing pattern-skeleton substitution,
    using that geometric clusters are preserved under left-quotienting by argument positions.
  \item[(inductive $\Rightarrow$ geometric)]
    algebraically, using a substitution lemma and 
    that the $\salgtre$-in\-ter\-pre\-ta\-tion contains the $\salgshi$-interpretation; \qedhere
  \end{list}
\end{proof}
\begin{example}
   $\seteq{\op{\pair{\syma{\mtrmvara{\bsymzer,\asymzer}}
               }{\msubstr{\amtrmvar}{\symb{\symc{\trmvar{\natone}},\trmvar{\nattwo}}}
               }
              }{\pair{\salgshi}{\salgtre}
              }
   }{\possco{\posone}{(\setdif{\op{\ip{\symb{\symc{\trmvar{\natone}},\trmvar{\nattwo}}
                                            }{[\trmvar{\natone},\trmvar{\nattwo} \mapsto
                                               \setemp,\setemp]                                               
                                            }}{\salgtre}
                                     }{\setstr{\edgpos{\posemp}}
                                     })}} \mathrel{\sseteq} $\hfill\\
   $\seteq{\possco{\posone}{(\setdif{\setstr{\edgpos{\posemp},\vtxpos{\posemp},
                                              \edgpos{\posone},\vtxpos{\posone}}
                                    }{\setstr{\edgpos{\posemp}}
                                    })}
   }{\seteq{\possco{\posone
                  }{\setstr{\vtxpos{\posemp},\edgpos{\posone},\vtxpos{\posone}}
                  }
          }{\setstr{\vtxpos{\posone},
                    \possco{\posone}{\edgpos{\posone}},
                    \possco{\posone}{\vtxpos{\posone}}}
          }}$.
\end{example}
Equipping clusters with the lattice operators
$\clstop$, $\clsbot$, $\sclsglb$, $\sclslub$, the 
theorem shows the join-irreducible elements of the refinement order
can be perceived as vertices (patterns comprising a single function symbol)
and edges (patterns comprising two function symbols), which can be seen
as justifying having both types of positions:
an edge is \emph{more} than the join of its endpoints.

\section{Critical peaks redefined}

We first redefine single/parallel/multi-steps in first-order term rewriting
\emph{by second-order means} as clusters with rule symbols~\cite{Tere:03} in patterns,
and next critical peaks via the clusters of its steps.

\begin{lemma}
  For a left-linear TRS, 
    $\arsa{\atrm}{\btrm}$ iff
    $\trmeq{\atrm}{\msubstra{\amtrmvar}{\alhs}{\amtrm}}$ and
    $\trmeq{\msubstra{\amtrmvar}{\arhs}{\amtrm}}{\btrm}$
    for some skeleton $\amtrm$ 
    and pattern substitution $\msubstr{\amtrmvar}{\alhs}$,
    for rule $\rulstr{\alhs}{\arhs}$ with $\alhs$ standard.
\end{lemma}
\begin{proof}
    If $\trmeq{\atrm}{\ctxa{\suba{\alhs}}}$ and
    $\trmeq{\ctxa{\suba{\arhs}}}{\btrm}$,
    then set 
    $\amtrm \isdefd \ctxa{\mtrmvara{\suba{\vec{\strmvar}}}}$
    and vice versa. 
\end{proof}
Turning rules into rule symbols, the lemma justifies representing a step
$\arsa{\atrm}{\btrm}$ as a cluster $\phi$ having the rule symbol as pattern
substitution.  We denote it by $\iarsa{\astp}{\atrm}{\btrm}$.
\begin{example}~\label{exa:step}
 The step $\arsa{\syma{\symap{\asym}{\syma{\asymzer}}}
         }{\syma{\symb{\syma{\asymzer},\syma{\asymzer}}}
         }$
 for the rule
  $\rula{\trmvar{\natone}} \hastype 
    \rulstr{\syma{\trmvar{\natone}}
          }{\symb{\trmvar{\natone},\trmvar{\natone}}
          }$
  can be represented by the cluster
  $\pair{\syma{\mtrmvara{\syma{\asymzer}}}
         }{\msubstr{\amtrmvar}{\rula{\trmvar{\natone}}}
         }$: projecting the rule $\arul$ in the substitution 
  to its left/right-hand side
  $\msubstr{\amtrmvar}{\syma{\trmvar{\natone}}}$/
  $\msubstr{\amtrmvar}{\symb{\trmvar{\natone},\trmvar{\natone}}}$
  yields the step.
\end{example}
We now use the lattice to measure the interaction between 
steps in peaks.\footnote{%
The lattice structure on clusters does, in itself, not give rise
to a lattice structure on rules/steps/reductions.}
Note that the top element $\clstop$ for an $\anat$-ary pattern $\alhs$
has shape
$\pair{\mtrmvara{\trmvar{\natone},\ldots,\trmvar{\anat}}
                     }{\msubstr{\amtrmvar}{\alhs}
                     }$.
\begin{definition}
  A local peak
  $\iarsinva{\astp}{\btrm}{\iarsa{\bstp}{\atrm}{\ctrm}}$ 
  is \emph{critical} if 
  $\clseq{\clslub{\astp}{\bstp}}{\clstop}$ with $\atrm$ standard,
  where we extend the refinement order to steps via their left-hand side.
\end{definition}
\begin{example}
   Consider the (standard) rules $\rulstr{\syma{\trmvar{\natone}}}{\trmvar{\natone}}$
   and $\rulstr{\syma{\trmvar{\natone}}}{\asymzer}$.
     \begin{itemize}
     \item
       $\arsinva{\trmvar{\natone}}{\arsa{\syma{\trmvar{\natone}}}{\asymzer}}$
       is critical since the union of the redex-patterns 
       is $\setstr{\vposemp}$,
       the set of all internal positions of $\syma{\trmvar{\natone}}$;
     \item
       $\arsinva{\bsymzer}{\arsa{\syma{\bsymzer}}{\asymzer}}$
       is not critical since the union of the redex-patterns
       is $\setstr{\vposemp}$, which is distinct from the set
       $\setstr{\vposemp,\eposone,\vposone}$ of all internal positions of $\syma{\bsymzer}$; and
     \item
       $\arsinva{\asymzer}{\arsa{\syma{\syma{\trmvar{\natone}}}}{\syma{\asymzer}}}$
       is not critical since the union of the redex-patterns 
       $\setstr{\vposemp,\vposone}$ misses the internal position $\eposone$ of
       $\syma{\syma{\trmvar{\natone}}}$.
     \end{itemize}
\end{example}
\begin{lemma}
  The definition of critical peak for a pair of rules is equivalent to the definitions
  found in the literature, up to 
  \emph{most generalness} (unifier or common instance),
  \emph{chiasmus} ($1$st rule--$2$nd rule vs.\ $2$nd rule--$1$st rule),
  \emph{order} (outer--inner vs.\ inner--outer),
  \emph{renaming} (variables in the peak), and
  \emph{triviality} (overlap of a rule with itself at the root).
\end{lemma}
\begin{proof}
  The definition of critical pair/peak varies along these
  parameters throughout the standard 
  literature~\cite{Huet:80,Ders:Joua:90,Baad:Nipk:98,Ohle:02,Tere:03}.
  The notions in the literature \emph{implement} our abstract notion.
\end{proof}
The above generalises to multi-steps~\cite{Tere:03} contracting a
number of (non-overlapping) redex-patterns at the same time.
We write $\arsdev$ ($\aiarsdev{\aStp}$) for multi-step 
(induced by cluster $\aStp$).
\begin{lemma}
  For a left-linear TRS, $\arsdeva{\atrm}{\btrm}$ iff
  $\trmeq{\atrm}{\msubstra{\vec{\amtrmvar}}{\vec{\alhs}}{\amtrm}}$ and
  $\trmeq{\msubstra{\vec{\amtrmvar}}{\vec{\arhs}}{\amtrm}}{\btrm}$, 
  for some skeleton $\amtrm$ and pattern substitution
  $\msubstr{\vec{\amtrmvar}}{\vec{\alhs}}$, 
  for rules $\overrightarrow{\rulstr{\alhs}{\arhs}}$ with $\vec{\alhs}$ standard.
\end{lemma}
\begin{definition}
  A peak
  $\iarsdevinva{\aStp}{\btrm}{\iarsdeva{\bStp}{\atrm}{\ctrm}}$ 
  is \emph{critical} if 
  $\clseq{\clslub{\aStp}{\bStp}}{\clstop}$ with $\atrm$ standard,
  where we extend the refinement order to multi-steps via their left-hand side.
\end{definition}
That is, the same concise ($5$ symbols) 
definition of critical peak as before allows us to capture
all the notions of parallel critical peaks and development critical peaks
(having definitions of up to $2$ pages),
due to Gramlich, Toyama, Okui, and Felgenhauer from the literature.
\begin{example}
  The following are critical peaks in our sense:
  \begin{itemize}
  \item
    the parallel (one--parallel) critical peak 
    $\arsinva{\csymzer}{\arspra{\syma{\asymzer,\asymzer}}{\syma{\bsymzer,\bsymzer}}}$
    for rules $\rulstr{\syma{\asymzer,\asymzer}}{\csymzer}$, $\rulstr{\asymzer}{\bsymzer}$;
  \item
    the development (one--multi) critical peak
    $\arsinva{\symb{\csymzer}
       }{\arsdeva{\symb{\syma{\asymzer}}
                }{\bsymzer
                }}$
    for rules $\rulstr{\syma{\asymzer}}{\csymzer}$, 
                  $\rulstr{\symb{\syma{\trmvar{\natone}}}}{\trmvar{\natone}}$,
                  $\rulstr{\asymzer}{\bsymzer}$; and
  \item
    the multi--multi critical peaks
        $\arsdevinva{\syma{\symap{\op{\bsym}{\anat}}{\trmvar{\natone}}}
          }{\arsdeva{\symap{\op{\asym}{\natsuc{\nattwo\anat}}}{\trmvar{\natone}}
                   }{\symap{\op{\bsym}{\anat}}{\syma{\trmvar{\natone}}}
                   }}$
    for $\rulstr{\syma{\syma{\trmvar{\natone}}}}{\symb{\trmvar{\natone}}}$.
 \end{itemize}
\end{example}
\begin{lemma}[Critical peak]
   If $\iarsdevinva{\aStp}{\btrm}{\iarsdeva{\bStp}{\atrm}{\ctrm}}$
   having more than one redex-pattern (in total),
   \begin{itemize}
   \item
      $\clseq{\clslub{\aStp}{\bStp}}{\clstop}$:
      the peak is a variable-substitution instance of a \emph{critical} peak;\footnote{%
Note that if $\clseq{\clslub{\acls}{\bcls}}{\clstop}$, then 
$\clsneq{\acls,\bcls}{\clsbot}$ iff $\clsneq{\clsglb{\acls}{\bcls}}{\clsbot}$
by connectedness and downward closedness of clusters.
}
      or
   \item
      $\clsneq{\clslub{\aStp}{\bStp}}{\clstop}$:
      $\trmeq{\aStp}{\substra{\atrmvar}{\aiStp{\natone}}{\aiStp{\natzer}}}$ and
      $\trmeq{\bStp}{\substra{\atrmvar}{\biStp{\natone}}{\biStp{\natzer}}}$, 
      for peaks
      $\iarsdevinva{\aiStp{\aidx}}{\bitrm{\aidx}}{\iarsdeva{\biStp{\aidx}}{\aitrm{\aidx}}{\citrm{\aidx}}}$
      with $\setin{\aidx}{\setstr{\natzer,\natone}}$, having \emph{smaller} skeletons.
   \end{itemize}
\end{lemma}
We sketch how the lemma allows to prove many 
confluence-by-critical-pair-analysis results by induction on the size
of the skeleton and splitting in these two cases.
\begin{example}
  \begin{itemize}
  \item
    the critical pair lemma for (left-linear) TRSs follows in the first
    case by the assumption that critical pairs are joinable, 
    and in the second case by the induction hypothesis (twice)
    and then using that reduction is closed under substitution
    to \emph{recompose};
  \item
    that orthogonal TRSs are confluent follows by proving that multi-steps
    have the diamond property, with the first case being trivial 
    since each critical peak is trivial by orthogonality, and 
    concluding in the second case by the induction hypothesis (twice)
    and then using that multi-steps are closed under substitution
    to \emph{recompose} the multi-steps; and
  \item
    that development-closed TRSs are confluent is proven by refining
    the proof of the previous item by an extra (outer) induction on
    the \emph{amount} of overlap between the multi-steps.
    For that it is essential that the refinement order is 
    a \emph{distributive} lattice, as that allows to express the
    amount of overlap between two multi-steps as 
    the sum of the amounts of overlap (the sizes of the meets) 
    between their constituting redex-patterns. 
  \end{itemize}
\end{example}
We expect the above extends to non-left-linear,\footnote{%
Although the definition of critical peak remains the same,
we then do not get a distributive lattice.} higher-order pattern, 
and graph rewriting.

\paragraph{Acknowledgements.}
This work has been partially supported by JSPS KAKENHI Grant Number 17K00011,
JSPS Core to Core Program, and Austrian Science Fund (FWF) project P27528.

\bibliographystyle{plain}
\bibliography{cpr-iwc2017}

\end{document}